\documentclass[english,10pt,letters]{IEEEtran}

\ifCLASSINFOpdf
\else
\fi
\usepackage{caption}
\usepackage{epsfig}
\usepackage{amsthm,amssymb,graphicx,graphicx,multirow,amsmath,color,mathtools}
\usepackage{tikz}
\usepackage{pgfplots}

\usepackage{epstopdf}
\usetikzlibrary{patterns}

\usepackage{verbatim}

\def\E{\mathbb{E}}

\def\ie{{\em i.e.}}

\def\scriptscriptstyle{}

\def\d{\mathrm{d}}

\def\nv{\sigma^2_n}
\def\iv{\sigma^2_i}

\def\ti{T_{\scriptscriptstyle i}}
\def\ts{T_{\scriptscriptstyle s}}

\def\sinr{\mathtt{SINR}}

\def\snr{\mathtt{SNR}}

\def\awgn{\mathtt{AWGN}}
\def\AWGN{\mathtt{AWGN}}
\def\pdf{\mathtt{PDF}}

\def\Po{\bar{P}}

{}
  
  \newtheorem{theorem}{Theorem}

\usepackage{subcaption}

\begin{document}
%
\title{Adaptive Modulation  with Impulsive Interference}
%
%
%
 
\author{\IEEEauthorblockN{Sudharsan Parthasarathy and  Radha Krishna Ganti}\\
\IEEEauthorblockA{Department of Electrical Engineering\\
Indian Institute of Technology Madras\\
Chennai, India 600036\\
\{sudharsan.p, rganti\}@ee.iitm.ac.in}
}

\maketitle

\begin{abstract}
In this letter, we analyze  power and rate adaptation in a point-to-point link with  Rayleigh fading and  impulsive interference. We  model the  impulsive interference as a Bernoulli-Gaussian random process.  Adaptation is used to maximize the average spectral efficiency by changing  power and rate of the transmission subject to an average power and  instantaneous probability of error constraints. Without impulsive interference, it is well known that water-filling is  optimal for block fading.  We provide two simple schemes that show that the  conventional water-filling algorithm is not optimal in an impulsive interference channel.
\end{abstract}

%
\IEEEpeerreviewmaketitle

\section{Introduction}
%
%
%
%

In current wireless systems, frequency planning and medium access control protocols help regulate co-channel interference.  Despite these protocols, due to the size, complexity of the networks and device imperfections, interference is a major impediment.  Sometimes, this interference is impulsive in nature \cite{ kapil}. For example, in a cellular system impulsive interference occurs because of partial loading and channel  feedback instability. Also, current wireless networks are interfered by other devices such as microwave ovens that are in close by vicinity and operate in the same frequency band \cite{kamerman} and the  interference generated is impulsive in nature. In such scenarios, Bernoulli-Gaussian arrival process is used to model the impulsive interference process \cite{capacity}. Impulsive noise exists in power line communications and has been analysed a lot in powerline literature in past decade. Zimmermann et al. proposed  the multipath model for a power line channel \cite{zimmer1}. They also analysed and modeled impulsive noise that occurs in power-line communications \cite{zimmer2}. Meng et al. gave a frequency domain approach to model and analysed the effects of noise  on broadband power-line communications \cite{meng}. In \cite{ma}, impulsive interference is analyzed in an OFDM system for power-line communications, but rate and power adaptation are not considered.  In  \cite{mawali},  rate and power adaptation are considered in the context of power-line communications. However the analysis is carried out with an average interference power assumption. 

Present day wireless systems  adapt transmission rate and transmit power to achieve a target probability of error, subject to an average power constraint. In a block fading channel with additive white Gaussian noise, the  signal-to-noise-ratio ($\snr$) is constant across a coherence time interval (coherence time interval is the duration when the channel gain remains constant). Hence a  water-filling  based solution (based on the $\snr$ fed back at the beginning of every coherence time interval) is optimal to adapt the power and the transmission rate \cite{andrea1} in such channels. 

Adaptive modulation in the presence of $\AWGN$  and persistent interference has been  studied in  \cite{taewon}, \cite{anders}. In a two user downlink interference channel, when both the channels experience Rayleigh fading, it was shown  in \cite{taewon} that binary power allocation  performs better than water-filling when ratio of the gain of the interfering path to the gain of the direct path is greater than a threshold. In \cite{anders} it was shown that though binary power allocation is sub-optimal for a system with multiple interfering links, it is preferred as it is computationally simple. However, there has been little understanding of how to adapt the power and modulation rate even in a point-to-point link when the interference is impulsive in nature.

\section{System Model}
\label{sec:basic}

\ifCLASSOPTIONtwocolumn
\begin{figure}[ht]
{
\centering
\begin{tikzpicture}[scale=1.5]
\draw [->] (-0.2,0) -- (4.2,0);
\draw [->] (0,-0.025) -- (0,.8);
\draw (0,-0.05) -- (0,0.05);
\draw (.5,-0.05) -- (.5,0.05);
\draw (1.5,-0.05) -- (1.5,0.05);
\draw (2.5,-0.05) -- (2.5,0.05);
\draw (3.5,-0.05) -- (3.5,0.05);
\draw (4,-0.05) -- (4,0.05);
\draw (1,-0.05) -- (1,0.05);
\draw (2,-0.05) -- (2,0.05);
\draw (3,-0.05) -- (3,0.05);

\node at (0,-0.15) {$\mathbf{0}$};
\node at (.5,-0.15) {$T_s$};
\node at (1,-0.15) {$2T_s$};
\node at (1.5,-0.15) {$3T_s$};
\node at (2,-0.15) {$\mathbf{T_c}$};
\node at (2.5,-0.15) {$5T_s$};
\node at (3,-0.15) {$6T_s$};
\node at (3.5,-0.15) {$7T_s$};
\node at (4,-0.15) {$\mathbf{2T_c}$};
;
\draw   [fill=pink](0,0) rectangle (.5,0.5); 
\draw  [fill=pink] (1.5,0) rectangle (2,0.5); 
 \draw  [fill=pink] (3,0) rectangle (3.5,0.5); 
 
%

\end{tikzpicture}
\caption{Bernoulli interference process: We consider $N=4$ symbols in a coherence time interval. The interference occurs as a block of $T_s$ with probability $p$. The channel estimate of the first symbol of coherence time interval is fed back based on which the transmitter adapts its power and transmission rate.}
\label{fig:bernoulli}
}
\end{figure}
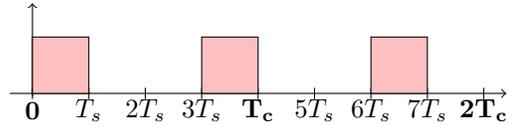

\else
\begin{figure}[ht]
{
\centering
\begin{tikzpicture}[scale=3]
\draw [->] (-0.2,0) -- (4.2,0);
\draw [->] (0,-0.025) -- (0,.8);
\draw (0,-0.05) -- (0,0.05);
\draw (.5,-0.05) -- (.5,0.05);
\draw (1.5,-0.05) -- (1.5,0.05);
\draw (2.5,-0.05) -- (2.5,0.05);
\draw (3.5,-0.05) -- (3.5,0.05);
\draw (4,-0.05) -- (4,0.05);
\draw (1,-0.05) -- (1,0.05);
\draw (2,-0.05) -- (2,0.05);
\draw (3,-0.05) -- (3,0.05);

\node at (0,-0.15) {$\mathbf{0}$};
\node at (.5,-0.15) {$T_s$};
\node at (1,-0.15) {$2T_s$};
\node at (1.5,-0.15) {$3T_s$};
\node at (2,-0.15) {$\mathbf{T_c}$};
\node at (2.5,-0.15) {$5T_s$};
\node at (3,-0.15) {$6T_s$};
\node at (3.5,-0.15) {$7T_s$};
\node at (4,-0.15) {$\mathbf{2T_c}$};
;
\draw   [fill=pink](0,0) rectangle (.5,0.5); 
\draw  [fill=pink] (1.5,0) rectangle (2,0.5); 
 \draw  [fill=pink] (3,0) rectangle (3.5,0.5); 
 
%

\end{tikzpicture}
\caption{Bernoulli interference process: We consider $N=4$ symbols in a coherence time interval. The interference occurs as a block of $T_s$ with probability $p$. The channel estimate of the first symbol of coherence time interval is fed back based on which the transmitter adapts its power and transmission rate.}
\label{fig:bernoulli}
}
\end{figure}
\fi

In this letter, we consider  a fading channel with impulsive interference. We assume that the duration of interference $\ti$ is same as that of symbol period $\ts$, \ie, $\ti =\ts$ as shown in Figure \ref{fig:bernoulli}.  We consider the following discrete time system  model for Bernoulli impulsive interference \cite{monisha},
\begin{equation}
y_m= h_{m} x_m +w_m+i_m, \quad  \forall m \geq 0.
\label{eqn:sys}
\end{equation}
Here $y_m$ is the received signal,  $x_m$ is the transmitted M-QAM symbol, $w_m$ denotes a zero mean Gaussian random variable  of variance $\nv$ and $i_m$  denotes the interference process.  The fading denoted by $h_m$  is an i.i.d process. We assume that the magnitude of the fading $|h_m|$ is Rayleigh distributed and hence $H_m=|h_m|^2$ is exponentially distributed. Without loss of generality, the mean of $H_m$ is assumed to be one. We will use $H$ to denote a generic instance of $H_m$ and it remains constant for the coherence time interval $T_c$.  The average signal power ($\overline{P}$) is also assumed to be one. The interference process is given by 
\begin{align*}
i_m=b_m g_m,
\end{align*}
where $b_m$ is  Bernoulli process and  $g_m$ is a sequence of i.i.d. Gaussian random variables of variance $\iv.$  The Bernoulli process $b_m$ is an i.i.d sequence of zeros and ones with $\mathbb{P}(b_m=0)=1-p$ and $\mathbb{P}(b_m=1)=p.$  The interference process is shown in Figure \ref{fig:bernoulli}.
The receiver feeds back $\sinr$ to the transmitter at the beginning of every coherence time interval. 

The $\sinr$ of symbols not affected by impulsive interference is $\frac{H \overline{P}}{\sigma_n^2}$. As $H$ is exponentially distributed of unit mean and average signal power ($\overline{P}$) is assumed to be one, $\sinr$ is exponentially distributed of mean \[\overline{\gamma}=\frac{1}{\sigma_n^2},\] and its PDF is denoted as $f_n(\gamma)$. Similarly $\sinr$ of symbols affected by impulsive interference is exponentially distributed  with mean \[\frac{1}{\sigma_n^2+\sigma_i^2}=\frac{\overline{\gamma}}{(1+\mu)},\]  where $\mu=\sigma_i^2 / \sigma_n^2$ and the corresponding PDF is denoted as $f_i(\gamma)$. So  the $\sinr$ distribution is  $(1-p)f_n(\gamma)$  $+p f_i(\gamma)$. 
\section{Adaptive modulation in the presence of impulsive interference}
\label{sec:realistic}

The objective is to adapt transmit rate and power instantaneously to maximize average spectral efficiency while keeping the  instantaneous probability of error below a threshold $P_b$ and satisfying the power constraint over a long horizon. Here we focus on M-QAM modulation.
In \cite{andrea1}, the probability of bit error with  M-QAM  constellation can be approximated as 
\begin{align}
P_b(\gamma) = c e^{\frac{-1.5 \gamma}{M-1}},
\label{eqn:error}
\end{align}
where $c$=0.2 and $\gamma$ is the $\sinr$, where $\gamma=\frac{H \overline{P}}{\sigma^2}$,  $H$ is the channel power, $\overline{P}$ is average symbol power and $\sigma^2$ is noise power.
Let $P(\gamma)$  and  $M(\gamma)$ denote the transmit power and modulation scheme chosen based on the $\sinr$  fed back at the beginning of the coherence time interval. On adaptation, the $\sinr$ is $\frac{H P(\gamma)}{\sigma^2}$ which is equivalent to $\frac{\gamma P(\gamma)}{\overline{P}}$ .  So the probability of error is  
\begin{equation}
P_b(\gamma)=c e^{\frac{-1.5 \frac{\gamma P(\gamma)}{\Po}}{ (M(\gamma)-1)}}.
\label{eqn:prob_adaptive}
\end{equation}
 Hence for a target $P_b$, rearranging (\ref{eqn:prob_adaptive})  gives \[M(\gamma)= 1+\frac{k\gamma P(\gamma)}{\Po},\]
where $k=-1.5/ \ln(\frac{P_b}{c})$. 
We want to maximize the average spectral efficiency  $\E_\gamma[\log_{2} M(\gamma)] $
subject to an average power constraint $\E_\gamma[ P(\gamma)] \leq \Po$. 

When $\sinr$ is constant in a coherence time interval, water-filling is the optimal power adaptation scheme \cite{andrea1} and is given by,
\begin{equation}
\frac{kP(\gamma)}{\Po} = \left\{ \begin{array}{rl} 
\frac{1}{\eta} - \frac{1}{\gamma } ; &\mbox{$\gamma \geq \eta$} \\
0 ; &\mbox{$\gamma < \eta$},
 \end{array} \right.
 \label{eqn:opti_power2}
\end{equation}
where the threshold $\eta$ is obtained by solving
 \begin{equation}
\int\limits_{\eta} ^\infty \left(\eta^{-1}-\gamma^{-1}\right)g(\gamma) \d \gamma=k.
\label{eqn:beta_n_solving}
\end{equation}
 The average spectral efficiency is given by
\begin{equation}
\label{eqn:Rwf}
\E_\gamma[\log_{2} M(\gamma)]=  \int\limits_{\eta} ^\infty \log_2 \left(\frac{\gamma}{\eta} \right) g(\gamma) \d \gamma,
\end{equation}
where $g(\gamma)$ denotes the $\pdf$ of the $\sinr$ in the coherence interval.

But optimal power and rate adaptation for a channel where $\sinr$ varies within the coherence time interval is not known in literature to the best of our knowledge. We apply the conventional water-filling algorithm discussed above in an impulsive interference system and compare with two other algorithms to show that conventional water-filling is not optimal with impulsive interference. 

In this letter `outage' is defined as the event when the probability of error of symbol transmitted is greater than the fixed probability of error constraint $P_b$. 

\subsection{Conventional water-filling}
The average spectral efficiency achieved by conventional water-filling in an impulsive interference system is not same as given by (\ref{eqn:Rwf}), where $g(\gamma)$=$(1-p)f_n(\gamma)$ + $p f_i(\gamma)$. This is because all the symbols in a coherence time interval   do not experience the same $\sinr$. Hence the rate and power that are adapted using the $\sinr$ of first symbol of a  coherence time interval can not satisfy the probability of error constraint across all the symbols in a  coherence time interval. 
\begin{theorem}
The average spectral efficiency achieved by conventional water-filling in a channel affected by impulsive interference is \[ (1-p)^2 \int\limits_{\eta} ^\infty \log_2 \left(\frac{\gamma}{\eta} \right) f_n(\gamma) \d \gamma + p \int\limits_{\eta} ^\infty \log_2 \left(\frac{\gamma}{\eta} \right) f_i(\gamma) \d \gamma.\]
\end{theorem}
\begin{proof}

Let the instantaneous $\sinr$ of symbols not affected by impulsive interference and affected by impulsive interference in a coherence time interval be denoted as $\gamma_n$ and $\gamma_i$ respectively. So $\gamma_i=\gamma_n / (1+\mu),$ where $\mu=\sigma_i^2 / \sigma_n^2.$ Let the first symbol of a coherence time interval not be affected by impulsive interference as illustrated  in Fig. \ref{fig:bernoulli} in the time interval  $[ T_c, 5 T_s]$. If water-filling is utilized,  the symbols in the coherence time interval ($[T_c, 2 T_c]$) are transmitted at  power $P(\gamma_n)$=  $\frac{\bar{P}}{k}(\frac{1}{\eta}-\frac{1}{\gamma_n})$ and rate $\log_2(M(\gamma_n))$=$\log_2(\gamma_n / \eta)$. So the probability of error of a symbol affected by impulsive interference in this block (for example $[6 T_s,7 T_s]$  in Fig. \ref{fig:bernoulli}) on adaptation is 
\begin{align*}
P_b(\gamma_i) &= c e^{\frac{-1.5 \gamma_i P(\gamma_n)}{\Po (M(\gamma_n)-1)}} \\ &=  c e^{\frac{-1.5 \gamma_i (\frac{1}{\eta} - \frac{1}{\gamma_n})}{k (\frac{\gamma_n}{\eta}-1)}} \\  &= ce^{\frac{-1.5}{k(1+\mu)}} > P_b.
\end{align*}
As $P_b(\gamma_i)> P_b,$ this results in outage. Hence, outage occurs when the first symbol in a coherence slot is not affected by impulsive interference and  symbols within that coherence time interval are affected by impulsive interference. As the arrival process is memoryless,  the outage probability is $P_{o}=p(1-p)$.  So the average spectral efficiency is  

\ifCLASSOPTIONtwocolumn
\begin{multline}
 ((1-p)-p(1-p)) \int\limits_{\eta} ^\infty \log_2 \left(\frac{\gamma}{\eta} \right) f_n(\gamma) \d \gamma +\\ p \int\limits_{\eta} ^\infty \log_2 \left(\frac{\gamma}{\eta} \right) f_i(\gamma) \d \gamma,
\label{eqn:water_fill}
\end{multline}
\else
\begin{equation}
 ((1-p)-p(1-p)) \int\limits_{\eta} ^\infty \log_2 \left(\frac{\gamma}{\eta} \right) f_n(\gamma) \d \gamma +\\ p \int\limits_{\eta} ^\infty \log_2 \left(\frac{\gamma}{\eta} \right) f_i(\gamma) \d \gamma,
\label{eqn:water_fill}
\end{equation}
\fi
 where $\eta$ is obtained from (\ref{eqn:beta_n_solving}).
\end{proof}

We now consider two other alternative transmission strategies: The first strategy is to be aggressive and assume that impulsive interference is not present \ie, all symbols are affected by $\AWGN$ of variance $\sigma_n^2$ and adapt power and rate accordingly. The second alternative is to be conservative and assume that impulsive interference exists always \ie, all symbols are affected by $\AWGN$ of variance $\sigma_n^2 +\sigma_i^2$ and base the adaptation on this assumption. 

\subsection{Aggressive water-filling}
In this method we neglect the impulsive interference and water-filling is used assuming only the noise density function.  However we use $H$ instead of the $\sinr$ for adaptive modulation\footnote{Due to space constraint we omit the proof that in a Rayleigh faded channel both the approaches give the same average spectral efficiency}.  Since the interference and noise are zero mean processes, the channel $H$ can be estimated by  the receiver by averaging out the sequence of training symbols that the transmitter sends at the beginning of a coherence interval. This channel value $H$ is fed back to the transmitter and is used to adapt power and rate for the coherence interval.


As  in Section \ref{sec:realistic}, on power and rate adaptation, the probability of error is $P_b = c e^{\frac{-1.5 H P(H) }{ \sigma_n^2(M(H)-1)}}$. The objective is to maximize  the average spectral efficiency $\E_H[\log_{2} M(H)]$, where $M(H)=1+\frac{k H P(H)}{\overline{P}},$ subject to an average power constraint \[\int\limits_{0} ^\infty P(H) f(H) \d H \leq \overline{P},\] where $k$=$-\frac{1.5 \overline{P}}{\sigma_n^2 \ln( \frac{P_b}{c})}$ and $f(H)$ is the $\pdf$ of $H$. The solution to the above problem is similar to   the water-filling solution  in Section \ref{sec:realistic}. The water-filling threshold $\alpha_n$ is obtained by solving $\int\limits_{\alpha_n}^\infty (\frac{1}{\alpha_n} - \frac{1}{H})f(H) \d H =k.$
Similar to the proof of Theorem 1,  it can be seen that by using the threshold $\alpha_n$ for symbols affected by impulsive interference, will result in  an outage as the probability of error constraint will be violated. Hence the average spectral efficiency achieved by aggressive water-filling is 
\begin{equation*}
R_n=\int\limits_{\alpha_n} ^\infty (1-p)\log_{2} \left(\frac{H}{\alpha_n} \right)f(H) \d H.
\end{equation*}
Using the fact that $f(H) = \exp(-H)$,  we obtain  \[ \int\limits_{\alpha_n} ^\infty \log_2 \left(\frac{H}{\alpha_n} \right) f(H) \d H = \log_2(e) \left( \frac{ e^{-\alpha_n }}{\alpha_n} -k \right), \] where $k$=$-\frac{1.5 \overline{P}}{\sigma_n^2 \ln(\frac{ P_b}{c})}$.
Hence 
\begin{equation}
R_n=(1-p) \log_2(e) \left( \frac{  e^{-\alpha_n  }}{\alpha_n} -k  \right).
\label{eqn:aggressive}
\end{equation}

\subsection{Conservative water-filling}
In this method, we assume that the symbols are always affected by impulsive interference. 
As we assume impulsive interference exists always, the objective is to satisfy an instantaneous probability of error constraint $P_b = c e^{\frac{-1.5 H \overline{P}}{(\sigma_n^2+\sigma_i^2)(M-1)}}$. As seen in previous Section,  water-filling threshold $\alpha_i$ can be found from average power constraint, 
$ \int\limits_{\alpha_i}^\infty \left(\frac{1}{\alpha_i}-\frac{1}{H}\right) f(H) \d H=k,$ where $k$=$-\frac{1.5 \overline{P}}{(\sigma_n^2 + \sigma_i^2)\ln(\frac{P_b}{c})}$. As we are assuming the worst case scenario, outage does not occur and hence the average spectral efficiency achieved by conservative water-filling is 
\begin{align}
R_i &=\int\limits_{\alpha_i} ^\infty \log_{2} \left(\frac{H}{\alpha_i} \right)f(H) \d H .\\ &=\log_2(e) \left( \frac{ e^{-\alpha_i }}{\alpha_i} -k \right),
\label{eqn:conservative}
\end{align}
where $k$=$-\frac{1.5 \overline{P}}{(\sigma_n^2 + \sigma_i^2)\ln(\frac{P_b}{c})}$.

%

\section{Numerical Results }
 For  simulation, we generate exponential random variables (square of Rayleigh) with mean $\snr$ for 1-$p$ fraction of total symbols (one lakh symbols). Similarly $p$ fraction of total symbols are generated as exponential random variables with mean $\frac{\snr}{1+\mu}$. Using these generated values,  adaptive modulation algorithms (conventional, conservative and aggressive water-filling) as described in the previous sections are evaluated to  verify that  the average spectral efficiency obtained is same as derived in equations (\ref{eqn:water_fill}), (\ref{eqn:aggressive}) and (\ref{eqn:conservative}).   In practice, only discrete rate modulation can be implemented.  In the subsequent figures, the curves obtained through evaluating the algorithms by generating the random variables are termed as ``WF simul"  and the curves obtained by evaluating the equations are termed  as ``WF theory".  The discrete rate modulation as it is only a minor modification of the scheme described in the previous Sections and is described in detail in \cite{andrea} for the $\awgn$ case. 
 
 The theoretical and simulated rate curves for conventional (\ref{eqn:water_fill}), aggressive (\ref{eqn:aggressive}) and conservative (\ref{eqn:conservative}) water-filling schemes are plotted in Figure \ref{fig:FPC_1} and \ref{fig:FPC_3} for low $\mu$, $\snr$ and high $\mu$, $\snr$ case.  We first observe that the simulations match the theoretical expressions.
 At low $p$, \ie, low impulsive interference, we see that the aggressive approach of neglecting   interference in computing the water-filling threshold is better. On the other hand, at high $p$, we see that it is better off to always assume persistent interference in computing the threshold. 
We also observe that the aggressive or conservative schemes are always better than conventional water-filling \ie, $\max(R_n, R_i) > R_w$ $\forall$ $p$ $\in$ $(0,1)$. As there are no analytical closed form expressions to calculate water-filling thresholds $\eta$, $\alpha_n$ and $\alpha_i$ it is tough to analytically compare the average spectral efficiency  achieved by conventional water-filling ($R_w$) with the ones achieved by aggressive ($R_n$) and conservative ($R_i$) water-filling. These graphs (as counter examples) prove that conventional water-filling is not the optimal scheme with impulsive interference even when the interference is uncorrelated across time. Aggressive and conservative schemes are optimal when $p=0$ and $p=1$ respectively. The intuition to use aggressive and conservative schemes was to check if these could be better than conventional water filling at low and high values of $p$ respectively. But an interesting observation is that for all values of $p$ either of these two schemes is better than conventional water-filling. 

Let $p_{th}$  be the intersection point ($p$) of the aggressive and conservative achievable average rates. Also we observed that $p_{th}$ increases with $\mu$. This is because when impulsive interference power is high, the achievable average rate due to conservative scheme (\ref{eqn:conservative}) is very low. Hence for a wide range of $p$ values it is better to experience outage and transmit aggressively than to transmit conservatively.

\section{Conclusion}
In this letter, we have derived the average spectral efficiency of a Bernoulli-Gaussian impulsive interference channel achieved by three different water-filling schemes subject to an average power and instantaneous probability of error constraint. We observe that the conventional water filling is sub-optimal. In fact, even simpler schemes that neglect  interference (for computing water filling threshold) outperform conventional water filling solution for a range of parameters.

\ifCLASSOPTIONtwocolumn
\begin{figure}
\centering
\begin{tikzpicture}
\draw[help lines] (0,0);
\begin{axis}[%
width=7cm,
height=6cm,
scale only axis,
xmin=0,
xmax=1,
xlabel={$p$},
ymin=0,
ymax=0.5,
ylabel={Average spectral efficiency bps/Hz},
legend style={at={(0.01,0.01)},anchor=south west,draw=black,fill=white,legend cell align=left}
]


\addplot [
color=blue,
solid,
mark=solid,
mark options={solid}
]
table[row sep=crcr]{
0   0.4842\\      
.1   0.4246\\                
.2   0.3707\\                
.3   0.3237\\   
.4    0.2845\\    
.5   0.2544\\                
.6    0.2349\\    
.7  0.2281\\    
.8  0.2360\\                               
.9 0.2612\\               
1    0.3064\\      
};
\addlegendentry{\footnotesize{Conventional WF theory}};

\addplot [
color=blue,
solid,
mark=*,
mark options={solid}
]
table[row sep=crcr]{
0    0.4841\\    
.1   0.4245\\                
.2   0.3709\\     
.3    0.3235\\                 
.4    0.2844\\     
.5   0.2543\\          
.6   0.2351\\   
.7   0.2280\\    
.8  0.2362\\                                    
.9 0.2612\\          
1     0.3064\\            
};
\addlegendentry{\footnotesize{Conventional WF simul.}};

\addplot [
color=red,
solid,
mark=solid,
mark options={solid}
]
table[row sep=crcr]{
0    0.4842\\    
.1  0.4357\\         
.2   0.3873\\   
.3    0.3389\\    
.4   0.2905\\    
.5   0.2421\\    
.6  0.1937\\                                        
.7    0.1452\\       
.8  0.0968\\                                    
.9  0.0484\\         
1     0\\            
};
\addlegendentry{\footnotesize{Aggressive WF theory}};

\addplot [
color=red,
solid,
mark=star,
mark options={solid}
]
table[row sep=crcr]{
0    0.4842\\   
.1   0.4358\\    
.2   0.3873\\    
.3    0.3389\\    
.4   0.2905\\              
.5   0.2422\\         
.6   0.1937\\                                        
.7    0.1453\\      
.8   0.0970\\                                     
.9  0.0484\\      
1     0\\            
};
\addlegendentry{\footnotesize{Aggressive WF simul.}};

\addplot [
color=black,
solid,
mark=solid,
mark options={solid}
]
table[row sep=crcr]{
0    .3064\\    
.1   .3064\\         
.2    .3064\\   
.3     .3064\\    
.4    .3064\\    
.5    .3064\\    
.6   .3064\\                                        
.7     .3064\\       
.8   .3064\\                                    
.9   .3064\\         
1      .3064\\            
};
\addlegendentry{\footnotesize{Conservative WF theory}};

\addplot [
color=black,
solid,
mark=triangle,
mark options={solid}
]
table[row sep=crcr]{
0   0.3063\\   
.1   0.3064\\         
.2     0.3065\\   
.3      0.3063\\   
.4     0.3064\\   
.5     0.3063\\           
.6    0.3066\\    
.7   0.3063\\   
.8    0.3065\\   
.9   0.3064\\   
1     0.3064\\                                                                                                 
};
\addlegendentry{\footnotesize{Conservative WF simul.}};

\end{axis}
\end{tikzpicture}%
\caption{$\mu= \sigma_i^2 / \sigma_n^2= 0$ dB, $\snr$ = 0 dB.  }
\label{fig:FPC_1}
\end{figure}

\begin{figure}
\centering
\begin{tikzpicture}
\draw[help lines] (0,0);
\begin{axis}[%
width=7cm,
height=6cm,
scale only axis,
xmin=0,
xmax=1,
xlabel={$p$},
ymin=0,
ymax=1.8,
ylabel={Average spectral efficiency bps/Hz},
legend style={at={(0.42,0.58)},anchor=south west,draw=black,fill=white,legend cell align=left}
]


\addplot [
color=blue,
solid,
mark=solid,
mark options={solid}
]
table[row sep=crcr]{
0    1.7524\\    
.1    1.4910\\                
.2   1.2432\\    
.3   1.0105\\    
.4     0.7940\\   
.5    0.5955\\                
.6     0.4173\\    
.7   0.2624\\                                      
.8   0.1354\\                          
.9  0.0524\\                       
1   0.0957\\  
};
\addlegendentry{\footnotesize{Conventional WF theory}};

\addplot [
color=blue,
solid,
mark=*,
mark options={solid}
]
table[row sep=crcr]{
0     1.7521\\   
.1     1.4909\\        
.2    1.2431\\   
.3      1.0105\\   
.4    0.7941\\    
.5   0.5953\\   
.6    0.4175\\                            
.7  0.2621\\     
.8  0.1355\\                        
.9 0.0524\\                
1  0.0958\\              
};
\addlegendentry{\footnotesize{Conventional WF simul.}};

\addplot [
color=red,
solid,
mark=solid,
mark options={solid}
]
table[row sep=crcr]{
0   1.7524\\   
.1  1.5772\\    
.2  1.4019\\         
.3     1.2267\\   
.4   1.0514\\        
.5    0.8762\\    
.6   0.7010\\   
.7   0.5257\\                                               
.8   0.3505\\                                    
.9  0.1752\\          
1     0\\            
};
\addlegendentry{\footnotesize{Aggressive WF theory}};

\addplot [
color=red,
solid,
mark=star,
mark options={solid}
]
table[row sep=crcr]{
0    1.7523\\  
.1      1.5776\\   
.2    1.4017\\    
.3     1.2267\\         
.4    1.0509\\                 
.5    0.8762\\               
.6    0.7011\\     
.7   0.5258\\                                      
.8   0.3501\\                                          
.9    0.1752\\        
1     0\\            
};
\addlegendentry{\footnotesize{Aggressive WF simul.}};

\addplot [
color=black,
solid,
mark=solid,
mark options={solid}
]
table[row sep=crcr]{
0    0.0957\\
.1   .0957\\         
.2    .0957\\   
.3     .0957\\    
.4    .0957\\    
.5    .0957\\    
.6   .0957\\                                        
.7     .0957\\       
.8   .0957\\                                    
.9   .0957\\         
1      .0957\\            
};
\addlegendentry{\footnotesize{Conservative WF theory}};

\addplot [
color=black,
solid,
mark=triangle,
mark options={solid}
]
table[row sep=crcr]{
0   0.0956\\    
.1    0.0957\\    
.2     0.0957\\    
.3      0.0957\\   
.4     0.0958\\   
.5       0.0957\\        
.6      0.0956\\           
.7    0.0956\\     
.8     0.0956\\   
.9      0.0957\\      
1        0.0958\\  
};
\addlegendentry{\footnotesize{Conservative WF simul.}};

\end{axis}
\end{tikzpicture}%
\caption{$\mu=\sigma_i^2 / \sigma_n^2=20$ dB, $\snr = 10$ dB}
\label{fig:FPC_3}
\end{figure}

\else

\begin{figure}
\centering
\begin{tikzpicture}
\draw[help lines] (0,0);
\begin{axis}[%
width=7cm,
height=6cm,
scale only axis,
xmin=0,
xmax=1,
xlabel={$p$},
ymin=0,
ymax=0.5,
ylabel={Average spectral efficiency bps/Hz},
legend style={at={(0.01,0.01)},anchor=south west,draw=black,fill=white,legend cell align=left}
]


\addplot [
color=blue,
solid,
mark=solid,
mark options={solid}
]
table[row sep=crcr]{
0   0.4842\\      
.1   0.4246\\                
.2   0.3707\\                
.3   0.3237\\   
.4    0.2845\\    
.5   0.2544\\                
.6    0.2349\\    
.7  0.2281\\    
.8  0.2360\\                               
.9 0.2612\\               
1    0.3064\\      
};
\addlegendentry{\scriptsize{Conventional WF theory}};

\addplot [
color=blue,
solid,
mark=*,
mark options={solid}
]
table[row sep=crcr]{
0    0.4841\\    
.1   0.4245\\                
.2   0.3709\\     
.3    0.3235\\                 
.4    0.2844\\     
.5   0.2543\\          
.6   0.2351\\   
.7   0.2280\\    
.8  0.2362\\                                    
.9 0.2612\\          
1     0.3064\\            
};
\addlegendentry{\scriptsize{Conventional WF simul.}};

\addplot [
color=red,
solid,
mark=solid,
mark options={solid}
]
table[row sep=crcr]{
0    0.4842\\    
.1  0.4357\\         
.2   0.3873\\   
.3    0.3389\\    
.4   0.2905\\    
.5   0.2421\\    
.6  0.1937\\                                        
.7    0.1452\\       
.8  0.0968\\                                    
.9  0.0484\\         
1     0\\            
};
\addlegendentry{\scriptsize{Aggressive WF theory}};

\addplot [
color=red,
solid,
mark=star,
mark options={solid}
]
table[row sep=crcr]{
0    0.4842\\   
.1   0.4358\\    
.2   0.3873\\    
.3    0.3389\\    
.4   0.2905\\              
.5   0.2422\\         
.6   0.1937\\                                        
.7    0.1453\\      
.8   0.0970\\                                     
.9  0.0484\\      
1     0\\            
};
\addlegendentry{\scriptsize{Aggressive WF simul.}};

\addplot [
color=black,
solid,
mark=solid,
mark options={solid}
]
table[row sep=crcr]{
0    .3064\\    
.1   .3064\\         
.2    .3064\\   
.3     .3064\\    
.4    .3064\\    
.5    .3064\\    
.6   .3064\\                                        
.7     .3064\\       
.8   .3064\\                                    
.9   .3064\\         
1      .3064\\            
};
\addlegendentry{\scriptsize{Conservative WF theory}};

\addplot [
color=black,
solid,
mark=triangle,
mark options={solid}
]
table[row sep=crcr]{
0   0.3063\\   
.1   0.3064\\         
.2     0.3065\\   
.3      0.3063\\   
.4     0.3064\\   
.5     0.3063\\           
.6    0.3066\\    
.7   0.3063\\   
.8    0.3065\\   
.9   0.3064\\   
1     0.3064\\                                                                                                 
};
\addlegendentry{\scriptsize{Conservative WF simul.}};

\end{axis}
\end{tikzpicture}%
\caption{$\mu= \sigma_i^2 / \sigma_n^2= 0$ dB, $\snr$ = 0 dB}
\label{fig:FPC_1}
\end{figure}

\begin{figure}
\centering
\begin{tikzpicture}
\draw[help lines] (0,0);
\begin{axis}[%
width=7cm,
height=6cm,
scale only axis,
xmin=0,
xmax=1,
xlabel={$p$},
ymin=0,
ymax=0.5,
ylabel={Average spectral efficiency bps/Hz},
legend style={at={(0.42,0.58)},anchor=south west,draw=black,fill=white,legend cell align=left}
]


\addplot [
color=blue,
solid,
mark=solid,
mark options={solid}
]
table[row sep=crcr]{
0    0.4842\\     
.1    0.4193\\                 
.2   0.3568\\           
.3   0.2967\\   
.4     0.2395\\   
.5    0.1856\\   
.6    0.1354\\   
.7   0.0897\\                                    
.8  0.0496\\                             
.9  0.0175\\                     
1    0.0155\\  
};
\addlegendentry{\scriptsize{Conventional WF theory}};

\addplot [
color=blue,
solid,
mark=*,
mark options={solid}
]
table[row sep=crcr]{
0    0.4841\\   
.1    0.4195\\           
.2    0.3568\\   
.3     0.2966\\    
.4    0.2395\\    
.5   0.1857\\                  
.6   0.1354\\         
.7   0.0896\\    
.8 0.0495\\                               
.9 0.0176\\             
1     0.0156\\              
};
\addlegendentry{\scriptsize{Conventional WF simul.}};

\addplot [
color=red,
solid,
mark=solid,
mark options={solid}
]
table[row sep=crcr]{
0   0.4842\\   
.1  0.4357\\    
.2  0.3873\\    
.3    0.3389\\    
.4  0.2905\\            
.5   0.2421\\    
.6  0.1937\\                                        
.7    0.1452\\       
.8  0.0968\\                                    
.9  0.0484\\         
1     0\\            
};
\addlegendentry{\scriptsize{Aggressive WF theory}};

\addplot [
color=red,
solid,
mark=star,
mark options={solid}
]
table[row sep=crcr]{
0    0.4841\\   
.1    0.4357\\   
.2    0.3873\\   
.3     0.3390\\     
.4    0.2905\\             
.5    0.2419\\                
.6    0.1936\\    
.7    0.1452\\                                             
.8   0.0968\\                                         
.9   0.0485\\         
1     0\\            
};
\addlegendentry{\scriptsize{Aggressive WF simul.}};

\addplot [
color=black,
solid,
mark=solid,
mark options={solid}
]
table[row sep=crcr]{
0    .0155\\    
.1   .0155\\         
.2    .0155\\   
.3     .0155\\    
.4    .0155\\    
.5    .0155\\    
.6   .0155\\                                        
.7     .0155\\       
.8   .0155\\                                    
.9   .0155\\         
1      .0155\\            
};
\addlegendentry{\scriptsize{Conservative WF theory}};

\addplot [
color=black,
solid,
mark=triangle,
mark options={solid}
]
table[row sep=crcr]{
0   0.0155\\    
.1    0.0155\\    
.2     0.0156\\    
.3      0.0155\\    
.4     0.0155\\           
.5      0.0155\\    
.6     0.0155\\          
.7   0.0156\\  
.8     0.0156\\ 
.9     0.0156\\  
1       0.0156\\      
};
\addlegendentry{\scriptsize{Conservative WF simul.}};

\end{axis}
\end{tikzpicture}%
\caption{$\mu=\sigma_i^2 / \sigma_n^2=20$ dB, $\snr = 0$ dB}
\label{fig:FPC_2}
\end{figure}

\begin{figure}
\centering
\begin{tikzpicture}
\draw[help lines] (0,0);
\begin{axis}[%
width=7cm,
height=6cm,
scale only axis,
xmin=0,
xmax=1,
xlabel={$p$},
ymin=0,
ymax=1.8,
ylabel={Average spectral efficiency bps/Hz},
legend style={at={(0.42,0.58)},anchor=south west,draw=black,fill=white,legend cell align=left}
]


\addplot [
color=blue,
solid,
mark=solid,
mark options={solid}
]
table[row sep=crcr]{
0    1.7524\\    
.1    1.4910\\                
.2   1.2432\\    
.3   1.0105\\    
.4     0.7940\\   
.5    0.5955\\                
.6     0.4173\\    
.7   0.2624\\                                      
.8   0.1354\\                          
.9  0.0524\\                       
1   0.0957\\  
};
\addlegendentry{\scriptsize{Conventional WF theory}};

\addplot [
color=blue,
solid,
mark=*,
mark options={solid}
]
table[row sep=crcr]{
0     1.7521\\   
.1     1.4909\\        
.2    1.2431\\   
.3      1.0105\\   
.4    0.7941\\    
.5   0.5953\\   
.6    0.4175\\                            
.7  0.2621\\     
.8  0.1355\\                        
.9 0.0524\\                
1  0.0958\\              
};
\addlegendentry{\scriptsize{Conventional WF simul.}};

\addplot [
color=red,
solid,
mark=solid,
mark options={solid}
]
table[row sep=crcr]{
0   1.7524\\   
.1  1.5772\\    
.2  1.4019\\         
.3     1.2267\\   
.4   1.0514\\        
.5    0.8762\\    
.6   0.7010\\   
.7   0.5257\\                                               
.8   0.3505\\                                    
.9  0.1752\\          
1     0\\            
};
\addlegendentry{\scriptsize{Aggressive WF theory}};

\addplot [
color=red,
solid,
mark=star,
mark options={solid}
]
table[row sep=crcr]{
0    1.7523\\  
.1      1.5776\\   
.2    1.4017\\    
.3     1.2267\\         
.4    1.0509\\                 
.5    0.8762\\               
.6    0.7011\\     
.7   0.5258\\                                      
.8   0.3501\\                                          
.9    0.1752\\        
1     0\\            
};
\addlegendentry{\scriptsize{Aggressive WF simul.}};

\addplot [
color=black,
solid,
mark=solid,
mark options={solid}
]
table[row sep=crcr]{
0    0.0957\\
.1   .0957\\         
.2    .0957\\   
.3     .0957\\    
.4    .0957\\    
.5    .0957\\    
.6   .0957\\                                        
.7     .0957\\       
.8   .0957\\                                    
.9   .0957\\         
1      .0957\\            
};
\addlegendentry{\scriptsize{Conservative WF theory}};

\addplot [
color=black,
solid,
mark=triangle,
mark options={solid}
]
table[row sep=crcr]{
0   0.0956\\    
.1    0.0957\\    
.2     0.0957\\    
.3      0.0957\\   
.4     0.0958\\   
.5       0.0957\\        
.6      0.0956\\           
.7    0.0956\\     
.8     0.0956\\   
.9      0.0957\\      
1        0.0958\\  
};
\addlegendentry{\scriptsize{Conservative WF simul.}};

\end{axis}
\end{tikzpicture}%
\caption{$\mu=\sigma_i^2 / \sigma_n^2=20$ dB, $\snr = 10$ dB}
\label{fig:FPC_3}
\end{figure}

\fi
\end{document}